\newif\ifpublic
\publictrue 

\documentclass[11pt]{article}

\usepackage{fullpage}
\ifpublic
\usepackage[disable]{todonotes}
\else
\usepackage[colorinlistoftodos]{todonotes}
\fi

\usepackage{amssymb,amsmath,amsthm,amsfonts,color,paralist}
\usepackage{mathpazo}
\allowdisplaybreaks
\numberwithin{equation}{section}
\usepackage[pagebackref,letterpaper=true,colorlinks=true,pdfpagemode=none,
linkcolor=blue,citecolor=blue,pdfstartview=FitH]{hyperref}
\usepackage[capitalize]{cleveref}

\renewcommand{\eqref}[1]{\hyperref[#1]{(\ref*{#1})}}

\usepackage{algorithm}
\usepackage{algorithmic}
\usepackage{xspace}

\newtheorem{theorem}{Theorem}
\numberwithin{theorem}{section}
\newtheorem{definition}[theorem]{Definition}
\newtheorem{lemma}[theorem]{Lemma}
\newtheorem{prop}[theorem]{Proposition}

\newtheorem*{remark}{Remark}

\newcommand{\R}{\mathbb{R}}

\newcommand{\suchthat}{\;:\;}

\newcommand{\indic}{\mathbb{I}}
\newcommand{\exx}{\mathcal{X}}

\newcommand{\abs}[1]{\left|#1\right|}
\newcommand{\size}[1]{\left|#1\right|}
\newcommand{\norm}[1]{\left\|#1\right\|}
\newcommand{\fnorm}[1]{\left\|#1\right\|_{F}}

\newcommand{\tr}[1]{\operatorname{Tr}\left(#1\right)}

\newcommand{\inner}[2]{\left\langle#1, #2\right\rangle}

\newcommand{\sr}[1]{\operatorname{sr}\left(#1\right)}
\newcommand{\rk}[1]{\operatorname{rank}\left (#1\right)}

\newcommand{\sparsestcut}{\textsc{Sparsest Cut}\xspace}
\newcommand{\uniformSparsestcut}{\textsc{Uniform Sparsest Cut}\xspace}

\title{Embedding approximately low-dimensional $\ell_2^2$ metrics into $\ell_1$}
\author{Amit Deshpande\footnote{Microsoft Research India, {\tt amitdesh@microsoft.com}} \and Prahladh Harsha\footnote{Tata Institute of Fundamental Research, {\tt prahladh@tifr.res.in}} \and Rakesh Venkat\footnote{Tata Institute of Fundamental Research, {\tt rakesh@tifr.res.in}}}

\date{}
\begin{document}

\begin{titlepage}
\thispagestyle{empty}

\maketitle

\begin{abstract}
  Goemans showed that any $n$ points $x_1, \dotsc x_n$ in
  $d$-dimensions satisfying $\ell_2^2$ triangle inequalities can be
  embedded into $\ell_{1}$, with worst-case distortion at most
  $\sqrt{d}$. We extend this to the case when the points are {\em approximately}
  low-dimensional, albeit with {\em average distortion} guarantees.  More precisely, we give an $\ell_{2}^{2}$-to-$\ell_{1}$
  embedding with average distortion at most the \emph{stable
    rank}, $\sr{M}$, of the matrix $M$ consisting of columns
  $\{x_i-x_j\}_{i<j}$. Average distortion embedding suffices for
  applications such as the \sparsestcut problem. Our embedding gives
  an approximation algorithm for the \sparsestcut problem on low
  threshold-rank graphs, where earlier work was inspired by Lasserre
  SDP hierarchy, and improves on a previous result of the first and third
  author [Deshpande and Venkat, {In \em Proc. 17th APPROX}, 2014]. Our ideas give a new perspective on
  $\ell_{2}^{2}$ metric, an alternate proof of Goemans' theorem, and a
  simpler proof for \emph{average} distortion $\sqrt{d}$. Furthermore, while the seminal result of Arora, Rao and Vazirani giving a $O(\sqrt{\log n})$ guarantee for $\uniformSparsestcut$ can be seen to  imply Goemans' theorem with \emph{average} distortion, our work opens up the possibility of proving such a result directly via a Goemans'-like theorem.
\end{abstract}
\end{titlepage}

\section{Introduction}
A finite metric space consists of a pair $(\exx, d)$, where $\exx$ is a finite set of points, and $d:\exx \times \exx \rightarrow \R_{\geq0}$ is a distance function on pairs of points in $\exx$. Finite metric spaces arise naturally in combinatorial optimization (e.g., the $\ell_1$ space in cut problems), and in practice (e.g., edit-distance between strings over some alphabet $\Sigma$). Since the input space may not be amenable to
efficient optimization, or may not admit efficient algorithms, one looks for \emph{embeddings} from these input spaces to easier spaces, while minimizing the \emph{distortion} incurred.  Given its importance, various aspects of such embeddings have been investigated such as dimension, distortion, efficient algorithms, and hardness results (refer to surveys~\cite{IndykM2004,Matousek2002,Linial2002} and references therein). In this paper, we provide better distortion guarantees for embedding \emph{approximately} low-dimensional points in the $\ell_2 ^2$-metric into $\ell_1$, and give applications to the \sparsestcut problem.

In the \sparsestcut problem, we are given graphs $C$, $D$ on the same vertex set $V$, with $\size{V} = n$, called the \emph{cost} and \emph{demand} graphs, respectively. They are specified by non-negative edge weights $c_{ij}, d_{ij} \geq 0$, for $i<j \in [n]$ and the \emph{(non-uniform) sparsest cut} problem, henceforth referred to as \sparsestcut, asks for a subset $S \subseteq V$ that minimizes
\[
\Phi(S) := \frac{\sum_{i<j} c_{ij} \abs{\indic_{S}(i) -
\indic_{S}(j)}}{\sum_{i<j} d_{ij} \abs{\indic_{S}(i) - \indic_{S}(j)}},
\]
where $\indic_{S}(i)$ is the indicator function giving $1$, if $i \in S$, and $0$, otherwise. We denote the optimum by $\Phi^{*} := \min_{S \subseteq V} \Phi(S)$. When the demand graph is a complete graph on $n$ vertices with uniform edge weights, the problem is then commonly referred to as the \uniformSparsestcut problem.

The best known (unconditional) approximation guarantee for the \uniformSparsestcut problem is $O(\sqrt{\log n})$, due to Arora, Rao and Vazirani \cite{AroraRV2009} (henceforth referred to as the ARV algorithm). Building on techniques in this work, Arora, Lee and Naor~\cite{AroraLN2008} give a $O(\sqrt{\log n} \log \log n)$ algorithm for non-uniform \sparsestcut. These  results come from a semi-definite programming (SDP) relaxation to produce solutions in the $\ell_2$-squared metric space, i.e., a set of vectors $\{x_i\}_{i\in V}$ in some high dimensional space that satisfy triangle inequality constraints on the squared distances in the following sense.
\[
\norm{x_{i}-x_{j}}_2^2 +\norm{x_{j}-x_{k}}_2^2\geq \norm{x_{i}-x_{k}}_2^{2} \qquad \forall~ i,j,k\in[n].
\]
Since the $\ell_{1}$ metric lies in the non-negative cone of cut (semi-)metrics, \cite{AroraRV2009} and \cite{AroraLN2008} round their solutions via low-distortion embeddings of the above $\ell_2^2$ solution into $\ell_1$ metric. Embeddings with low \emph{average-distortion} suffice for applications to the \sparsestcut problem.

Any $n$ points satisfying $\ell_{2}^{2}$ triangle inequalities make only acute angles among themselves, and therefore must lie in $\Omega(\log n)$ dimensions (Chapter 15, \cite{AignerZ2009}). However, for low \emph{threshold-rank} graphs, or more generally, when the $r$-th smallest generalized eigenvalue of the cost and demand graphs satisfies $\lambda_r(C,D) \gg \Phi_{SDP}$, the above SDP solution is known to be \emph{approximately} low-dimensional, that is, the span of its top $r$ eigenvectors contains nearly all of its total eigenmass (implicit in \cite{GuruswamiS2013}). Moreover, it can be embedded into $\ell_{1}$ using solutions of higher-levels of the  Lasserre SDP hierarchy to obtain a PTAS-like approximation guarantee \cite{GuruswamiS2013}. This motivates the quest for finding more efficient embeddings of low-dimensional or \emph{approximately} low-dimensional $\ell_{2}^{2}$ metrics into $\ell_{1}$.

Goemans (unpublished, appears in~\cite{MagenM2008}) showed that if the points satisfying $\ell_{2}^{2}$ triangle inequalities lie in $d$ dimensions, then they can be embedded into $\ell_2$ (and hence into $\ell_1$, since there is an isometry from $\ell_2 \hookrightarrow \ell_1$ \cite{Matousek2002}) with $\sqrt{d}$ distortion.
\begin{theorem}[{Goemans~\cite[Appendix~B]{MagenM2008}}] \label{thm:goemans-intro}
Let $x_{1}, x_{2}, \dotsc, x_{n} \in \R^{d}$ be $n$ points satisfying $\ell_{2}^{2}$
triangle inequalities. Then there exists an $\ell_{2}^{2}\hookrightarrow\ell_{2}$
embedding $x_{i} \mapsto f(x_{i})$ with distortion $\sqrt{d}$, that is,
\[
\frac{1}{\sqrt{d}}~ \norm{x_{i} - x_{j}}_{2}^{2} \leq \norm{f(x_{i}) -
f(x_{j})}_{2} \leq \norm{x_{i} - x_{j}}_{2}^{2}, \quad \forall~ i, j\in
V.
\]
\end{theorem}
%

\paragraph{Comparison of Goemans and ARV:} Since $n$ points satisfying $\ell_{2}^{2}$ triangle inequalities must lie in $d = \Omega(\log n)$ dimensions (Chapter 15, \cite{AignerZ2009}), the ARV algorithm \cite{AroraRV2009} implies an $\ell_{2}^{2}\hookrightarrow\ell_{1}$ embedding with \emph{average} distortion $O(\sqrt{d})$, and Arora-Lee-Naor \cite{AroraLN2008} improve it to $\tilde{O}(\sqrt{d})$ worst-case distortion. In the other direction, is it possible to extend \cref{thm:goemans-intro} to give ARV-like guarantees? Here are two immediate ideas that come to mind.
\begin{itemize}
\item Combine \cref{thm:goemans-intro} with a dimension reduction to $O(\log n)$ dimensions for $\ell_{2}^{2}$ metrics, similar to the Johnson-Lindenstrauss lemma for $\ell_{2}$. Such a dimension reduction for $\ell_{2}^{2}$ that approximately preserves all pairwise $\ell_{2}^{2}$ distances is ruled out by Magen and Moharrami~\cite{MagenM2008}, although their proof does not rule out dimension reduction for \emph{average} distortion.
\item Extend \cref{thm:goemans-intro} to work with approximate $\ell_{2}^{2}$ triangle inequalities, and then combine it with the Johnson-Lindenstrauss lemma. The Johnson-Lindenstrauss lemma, when applied to points satisfying $\ell_{2}^{2}$ triangle inequalities, preserves their $\ell_2^2$ triangle inequalities only approximately. That is, the points after the Johnson-Lindenstrauss random projection satisfy
    \[
    \norm{x_{i}-x_{j}}_2^2 + \norm{x_{j}-x_{k}}_2^2\geq \left(1- O(\epsilon)\right) \norm{x_{i}-x_{k}}_2^{2} \qquad \forall~ i, j, k \in [n].
    \]
    In fact, a generalization of \cref{thm:goemans-intro} that accommodates approximate $\ell_{2}^{2}$ triangle inequalities does hold, but its only proof (due to Trevisan~\cite{Trevisan15}) that we are aware of uses the technical core of the analysis of the ARV algorithm.
\end{itemize}


Here we seek a robust generalization of Goemans' theorem that avoids the above caveats. Our generalization of Goemans' theorem uses \emph{average} distortion instead of worst-case. It is also robust in the sense that it works with \emph{approximate} dimension instead of the actual dimension. Such a robust generalization opens up another possible approach to the general \sparsestcut problem: reduce the \emph{approximate} dimension while preserving the pairwise distances on \emph{average}, and then apply the robust version of Goemans' theorem. Moreover, our definition of the \emph{approximate} dimension is spectral, and our results can be easily compared to those of Guruswami-Sinop \cite{GuruswamiS2013} on Lasserre SDP hierarchies and Kwok et al.~\cite{KwokLLGT2013} on higher order Cheeger inequalities (see \cref{sec:results,subsec:related-works} for comparisons).


\subsection{Our Results}\label{sec:results}
We prove a robust version of Goemans' theorem, when the points $x_{1}, x_{2}, \dotsc, x_{n}$ are only \emph{approximately} low-dimensional. We quantify this \emph{approximate} dimension by the \emph{stable rank} of the difference matrix $M \in \R^{d \times {n \choose 2}}$ having columns $\{x_i - x_j\}_{i<j}$. Stable rank of the difference matrix is a natural choice because \begin{inparaenum}[(a)] \item stable rank is a continuous proxy for rank or dimension arising naturally in many applications \cite{BourgainT1987,Tropp2009}, \item the difference matrix $M$ is invariant under any shift of origin, and \item the difference matrix of the SDP solution for the \sparsestcut problem on \emph{low threshold-rank} graphs indeed has low stable rank (implicit in \cite{GuruswamiS2013}). \end{inparaenum}


\begin{definition}[Stable Rank]
Given $x_1, \ldots, x_n \in \R^d$, let $M \in \R^{d \times {n \choose 2}}$ be the matrix with columns
$\{x_i - x_j\}_{i<j}$. The stable rank of the points is defined as the
stable rank of $M$, given by~ $\sr{M} := \fnorm{M}^2 / \norm{M}_2^2$, where $\fnorm{M}$ and $\norm{M}_2$ are the Frobenius and spectral norm
of $M$ respectively.
\end{definition}
Note that $\sr{M} \leq \rk{M} \leq d$, when the points $x_{1}, x_{2}, \dotsc, x_{n} \in \R^{d}$. Our robust generalization of Goemans' theorem is as follows.
\begin{theorem}[Embedding almost low-dimensional vectors]\label{thm:stable-rank}
Let $x_{1}, x_{2}, \dotsc, x_{n} \in \R^{d}$ be $n$ points satisfying
$\ell_{2}^{2}$ triangle inequalities. Then there exists an
$\ell_{2}^{2}\hookrightarrow\ell_{2}$ embedding $x_{i} \mapsto h(x_{i})$ with
\emph{average} distortion bounded by the stable rank of $M$, that is,
\begin{align*}
\norm{h(x_{i}) - h(x_{j})}_{2} & \leq \norm{x_{i} - x_{j}}_{2}^{2}, \quad
\forall i, j\in V, \\
\text{and} \qquad \frac{1}{\sr{M}}~ \sum_{i<j} \norm{x_{i} - x_{j}}_{2}^{2} &
\leq \sum_{i<j} \norm{h(x_{i}) - h(x_{j})}_{2}.
\end{align*}
\end{theorem}

Our proof technique gives a new perspective on $\ell_{2}^{2}$ metric, an alternate proof of Goemans' theorem, and a simpler proof for average distortion $\sqrt{d}$ based on a squared-length distribution (see \cref{sec:goemans}). Also, the result can be quantitatively compared to guarantees given by higher-order Cheeger inequalities \cite{KwokLLGT2013}; we discuss this in more detail at the end of this section. While most known embeddings from $\ell_2 ^2 $ to $\ell_1$ are Frechet embeddings, our embedding is projective (similar in spirit to \cite{GuruswamiS2013, DeshpandeV2014}). To obtain a truly robust version of Goemans' theorem
quantitatively, one might ask if the dependence on $\sr{M}$ in the
above theorem can be improved from $\sr{M}$ to
$\sqrt{\sr{M}}$.

\Cref{thm:stable-rank} immediately implies an $\sr{M}$-approximation to the \uniformSparsestcut
problem. In fact, with a slight modification, we obtain a similar
result for the general \sparsestcut problem (see theorem below).

\begin{theorem} \label{thm:SparsestCut}
There is an $r/\delta$-approximation algorithm for \sparsestcut instances $C,D$ satisfying $\lambda_r(C,D) \geq \Phi_{SDP}
/(1-\delta)$, where $\lambda_r(C,D)$ is the $r$-th smallest generalized eigenvalue of the cost and demand graphs.
\end{theorem}
The precondition on $\lambda_{r}(C, D)$ is the same as in previous works~\cite{GuruswamiS2013,DeshpandeV2014}, and we improve the $O(r/\delta^{2})$-approximation of \cite{DeshpandeV2014} by a factor of $1/\delta$. Our proof follows from the robust version of Goemans' embedding into $\ell_{2}$ whereas these previous works gave embeddings directly into $\ell_{1}$ by either using higher levels of Lasserre explicitly \cite{GuruswamiS2013} or using only the basic SDP solution but inspired by the properties of Lasserre vectors \cite{DeshpandeV2014}.




\subsection{Related work} \label{subsec:related-works}
We recall that the best known upper bound for the worst-case distortion of
embedding $\ell_2^2 \hookrightarrow \ell_1$ is $O(\sqrt{\log n} \cdot \log \log n)$~\cite{AroraRV2009,AroraLN2008}, 
while the best known lower bound is $(\log n)^{\Omega(1)}$ for worst-case distortion~\cite{CheegerKN2009}, and
$\exp(\Omega(\sqrt{\log \log n}))$ for average distortion~\cite{KaneM2013}. Guarantees to \sparsestcut on \emph{low threshold-rank} graphs were obtained using higher levels of the Lasserre hierarchy for SDPs~\cite{BarakRS2011,GuruswamiS2013}. In contrast, a previous work of the first and third author~\cite{DeshpandeV2014} showed weaker guarantees, but using just the basic SDP relaxation. Oveis Gharan and Trevisan \cite{GharanT2013} also give a rounding algorithm for the basic SDP relaxation on low-threshold rank graphs, but require a stricter pre-condition on the eigenvalues ($\lambda_r \gg \log^{2.5} r \cdot \Phi(G)$), and leverage it to give a stronger $O(\sqrt{\log r})$-approximation guarantees. Their improvement comes from a new structure theorem on the SDP solutions of low threshold-rank graphs being clustered, and using the techniques in ARV for analysis.

Kwok~et al.~\cite{KwokLLGT2013} showed that a better analysis of Cheeger's inequality gives a $O(r \cdot \sqrt{d/\lambda_{r}})$ approximation to the sparsest cut in $d$-regular graphs. In particular, when $\lambda_r \geq \epsilon d$, this gives a $O(r/\sqrt \epsilon)$ approximation for the
\uniformSparsestcut problem. In this regime, our result  gives a slightly better approximation: Assuming $\lambda_r \geq \epsilon d$, if $\phi_{SDP} \leq \epsilon d/100 n$ then $\lambda_r \geq 100 \phi_{SDP}$ yielding an $O(r)$ approximation by \cref{thm:SparsestCut}. Otherwise, if $\phi_{SDP} \geq \epsilon d/100n$, then running a Cheeger rounding on the SDP solution would itself give a cut of sparsity $O(d\sqrt{\epsilon}/n) \leq \phi_{SDP}/\sqrt{\epsilon}$. Thus, the better of our rounding algorithm and a Cheeger rounding on the SDP solution gives a $\max \left\{O(r), 1/\sqrt{\epsilon}\right\}$-approximation to the \uniformSparsestcut problem.

Further, while the Kwok~et al. result is tight with respect to the spectral solution, our approach allows for an improvement in terms of the dependence on $r$ to $\sqrt{r}$, since it uses the SDP relaxation rather than a spectral solution.
\section{Preliminaries and Notation}\label{sec:prelims}


\paragraph*{Sets, Matrices, Vectors:}
We use $[n]=\{1,\ldots,n\}$. For a matrix $X \in \mathbb{R}^{d\times d}$, we say $X\succeq 0$ or $X$ is positive-semidefinite (psd) if $y^TXy \geq 0$ for all $y\in \mathbb{R}^d$. The Gram-matrix of a matrix $M \in \mathbb{R}^{d_1 \times d_2}$ is the matrix $M^T M$, which is psd.

Every matrix $M$ has a singular value decomposition $M=\sum_i \sigma_i u_i v_i^T = UDV^{T}$. Here, the matrices $U, V$ are Unitary, and $D$ is the diagonal matrix of the singular values $\sigma_1 \geq \sigma_2 \geq \ldots \geq \sigma_n$, in non-increasing order. When not clear from context, we denote the singular values of $M$  by $\sigma_i(M)$.

The \emph{Frobenius} norm of $M$ is given by $\fnorm{M} := \sqrt{\sum_i \sigma_i^2(M)} = \sqrt{\sum_{i\in[d_1],j\in[d_2]}M(i,j)^2}$. In our analysis, we will sometimes view a matrix $M$ as a collection of its columns viewed as vectors; $M=(m_j)_{j\in[d_2]}$. In this case, $\fnorm{M}^2=\sum_j \norm{m_j}_2^2$. The spectral norm of $M$ is $\norm{M}_2 := \sigma_1$.

\paragraph*{Rank and Stable Rank:}
The rank of the matrix $M$ (denoted by $\rk{M}$) is the number of non-zero singular values. Recall that the \emph{stable rank} of the matrix $M$, $\sr{M} = \frac{\fnorm{M}^2}{\sigma_1(M)^2}$. Note that $\rk{M} \geq \sr{M}$.

\paragraph*{Metric spaces and embeddings:}
For our purposes, a (semi-)metric space $(\mathcal{X}, d)$  consists of a finite set of points $\mathcal{X}=\{x_1 , x_2, \ldots, x_n\}$ and a distance function $d: \mathcal{X} \times \mathcal{X} \mapsto \R_{\geq0}$ satisfying the following three conditions:
\begin{enumerate}
\item $d(x,x)=0$, $\forall x\in \mathcal{X}$.
\item $d(x,y)=d(y,x)$.
\item (Triangle inequality) $d(x,y)+d(y,z) \geq d(x,z)$.
\end{enumerate}
An \emph{embedding} from a metric space $(\mathcal{X}, d)$ to a metric space $(\mathcal{Y}, d')$ is a mapping $f: \mathcal{X} \rightarrow \mathcal{Y}$. The embedding is called a \emph{contraction}, if
$$d'(f(x_i), f(x_j)) \leq d(x_i, x_j), \qquad\forall x_i, x_j \in \mathcal{X}.$$
For convenience, we will only deal with contractive mappings in this paper.
A contractive mapping is said to have (worst-case) distortion $\Delta$, if $$\sup_{i,j} \frac{d(x_i,x_j)}{d'(f(x_i),f(x_j))} \leq \Delta.$$ It is said to have \emph{average} distortion $\beta$, if
\[
\frac{\sum_{i<j}d(x_i,x_j)}{\sum_{i<j}d(f(x_i),f(x_j))} \leq \beta.
\]
Note that a mapping with worst-case distortion $\Delta$ also has average distortion $\Delta$, but not necessarily vice-versa.
\paragraph*{The $\ell_2^2$  space:}
A set of points $\{x_1, x_2, \ldots, x_n\}\in \R^d$ are said to satisfy $\ell_2^2$ triangle inequality constraints, or said to be in $\ell_2^2$ space, if it holds that  $$\norm{x_{i}-x_{j}}_2^2 +\norm{x_{j}-x_{k}}_2^2\geq \norm{x_{i}-x_{k}}_2^{2}  \qquad \forall i,j,k\in[n].$$  These satisfy the triangle inequalities on the \emph{squares} of their $\ell_2$ distances. The corresponding metric space is $(\mathcal{X},d)$, where $d(i,j):=\norm{x_i-x_j}_2^2$.

\paragraph*{Graphs and Laplacians:}
All graphs will be defined on a vertex set $V$ of size $n$. The vertices will usually be referred to by indices $i,j,k,l \in[n]$. Given a graph with weights on pairs $W:{V \choose 2}\mapsto \mathbb{R}^+$, the graph Laplacian matrix is defined as:
\begin{align*}
 L_W(i,j) :=
 \begin{cases}
 -W(i,j) &\quad \text{if $i\neq j$} \\
 \sum_k W(i,k) &\quad \text{if $i=j$}. \\
 \end{cases}
\end{align*}

\paragraph*{\sparsestcut SDP:}\label{sdp:ARVSDP}
The SDP we use for \sparsestcut on the vertex set $V$ with costs and demands $c_{ij}, d_{kl} \geq 0$  and corresponding cost and demand graphs $C:{V \choose 2}\mapsto \mathbb{R}^+$ and $D:{V\choose 2}\mapsto \mathbb{R}^+$, is effectively the following:
\begin{align*}
\textbf{SDP:}\quad \Phi_{SDP} & := \min \sum_{i<j} c_{ij} \norm{x_{i}-x_{j}}_2^2 \\
&\text{subject to}\begin{cases}\norm{x_{i}-x_{j}}_2^2 +\norm{x_{j}-x_{k}}_2^2\geq \norm{x_{i}-x_{k}}_2^{2} & \forall i,j,k\in[n].\\
{\sum_{k<l} d_{kl} \norm{x_{k}-x_{l}}_2^2} =1.&
\end{cases}
\end{align*}
Note that the solution to the above SDP is in $ \ell_2 ^2 $ space.

\paragraph*{$\ell_1$ embeddings and cuts:}
Since $\ell_1$ metrics are exactly the cone of cut-metrics, it follows from the previous discussion on embeddings, that producing an embedding of the SDP solutions $\mathcal{X}=\{x_1,\ldots. x_n\}$ in $\ell_2^2$ space to $\ell_1$ space with distortion $\alpha$ would give an $\alpha$-approximation to \sparsestcut. Producing one with \emph{average} distortion $\alpha$ would give an $\alpha$-approximation to \uniformSparsestcut. Furthermore, since $\ell_2$ embeds isometrically (distortion $1$) into $\ell_1$, it suffices to show embeddings into $\ell_2$ for the above purposes.

\paragraph*{Key Lemma:}
The following lemma about $\ell_2^2$ spaces was observed by Deshpande and Venkat~\cite{DeshpandeV2014}. We will reuse this in the rest of the paper.
\begin{lemma}[{\cite[Proposition~1.3]{DeshpandeV2014}}]\label{lemma:key}
Let $x_{1}, x_{2}, \dotsc, x_{n}$ be $n$ points satisfying $\ell_{2}^{2}$ triangle inequalities. Then
\[
\inner{x_{i} - x_{j}}{\frac{x_{k} - x_{l}}{\norm{x_{k} - x_{l}}_{2}}}^{2} \leq \abs{\inner{x_{i} - x_{j}}{x_{k} - x_{l}}} \leq \norm{x_{i} - x_{j}}_{2}^{2}, \quad \forall i, j, k, l\in V.
\]
\end{lemma}
An immediate consequence of this lemma is that we can show that a large class of naturally defined $\ell_2^2\hookrightarrow\ell_2$ embeddings are contractions.
\begin{lemma}[Contraction]\label{lem:contraction}
Let $x_{1}, x_{2}, \dotsc, x_{n}$ be $n$ points satisfying $\ell_{2}^{2}$ triangle inequalities. For any probability distribution $\{p_{kl}\}_{k<l}$, let $P$ be the symmetric psd matrix defined as $P := \sum_{k<l} p_{kl}~ (x_{k} - x_{l}) (x_{k} - x_{l})^{T}$. Then the $\ell_{2}^{2}\hookrightarrow\ell_{2}$ embedding given by $x_{i} \mapsto P^{1/2} x_{i}$ is a contraction, that is,
\[
\norm{P^{1/2} (x_{i} - x_{j})}_{2} \leq \norm{x_{i} - x_{j}}_{2}^{2}, \quad \forall i, j\in V.
\]
\end{lemma}
\begin{proof}
The following holds for all $i,j$:
\begin{align*}
\norm{P^{1/2} (x_{i} - x_{j})}_{2} & = \left((x_{i} - x_{j})^{T} P (x_{i} - x_{j})\right)^{1/2} \\
& = \left(\sum_{k<l} p_{kl} \inner{x_{i} - x_{j}}{x_{k} - x_{l}}^{2}\right)^{1/2} \\
& \leq \left(\sum_{k<l} p_{kl}~ \norm{x_{i} - x_{j}}_{2}^{4}\right)^{1/2} &[\text{By \cref{lemma:key}}]\\
& = \norm{x_{i} - x_{j}}_{2}^{2}. &[\text{Since $\sum_{k<l} p_{kl} = 1$}]
\end{align*}
\end{proof}
%
%

\section{Embedding almost low-dimensional vectors} \label{sec:stable-rank}
We now prove the robust version of Goemans' theorem in terms of stable rank. We give two proofs, and show an application to round solutions to \sparsestcut on low-threshold-rank graphs. As before, given a set of points $x_1, \dotsc, x_n$ in $\R^d$, define their difference matrix $M \in \R^{d\times {n \choose 2}}$ as the matrix with columns as $\{x_i-x_j\}_{i<j}$.
\begin{proof}[Proof of {\cref{thm:stable-rank}}]
Let $u$ and $v$ be the top left and right singular vector of $M$, respectively, and $\sigma_{1}\leq\sigma_{2}\leq\dotsc\leq\sigma_{d}$ be the singular values of $M$. Then $Mv = \sigma_{1} u$, or in other words, $\sigma_{1} u = \sum_{k<l} v_{kl} (x_{k} - x_{l})$. Now consider the embedding $x_{i} \mapsto h(x_{i}) = P^{1/2} x_{i}$, where the probability distribution $p_{kl} \propto \abs{v_{kl}}$, that is
\[
P = \sum_{k<l} \frac{\abs{v_{kl}}}{\norm{v}_{1}}~ (x_{k} - x_{l}) (x_{k} - x_{l})^{T}.
\]
This embedding is a contraction by \Cref{lem:contraction}. Now let's bound its average distortion.
\begin{align*}
\sum_{i<j} \norm{h(x_{i}) - h(x_{j})}_{2} & = \sum_{i<j} \norm{P^{1/2} (x_{i} - x_{j})}_{2} \\
& = \sum_{i<j} \left((x_{i} - x_{j})^{T} P (x_{i} - x_{j})\right)^{1/2} \\
& = \sum_{i<j} \left(\sum_{k<l} \frac{\abs{v_{kl}}}{\norm{v}_{1}}~ \inner{x_{i} - x_{j}}{x_{k} - x_{l}}^{2}\right)^{1/2} \\
& \geq \sum_{i<j} \sum_{k<l} \frac{\abs{v_{kl}}}{\norm{v}_{1}}~ \abs{\inner{x_{i} - x_{j}}{x_{k} - x_{l}}} &[\text{By Jensen's inequality}] \\
& \geq \sum_{i<j} \frac{1}{\norm{v}_{1}}~ \abs{\inner{x_{i} - x_{j}}{\sum_{k<l} v_{kl} (x_{k} - x_{l})}} &[\text{By triangle inequality}] \\
& = \frac{1}{\norm{v}_{1}}~ \sum_{i<j} \abs{\inner{x_{i} - x_{j}}{\sigma_{1} u}} \\
& = \frac{1}{\norm{v}_{1}}~ \sum_{i<j} \sigma_{1}^{2} \abs{v_{ij}} \\
& = \sigma_{1}^{2}  = \frac{\fnorm{M}^2}{\sr{M}}\\
& = \frac{1}{\sr{M}}~ \sum_{i<j} \norm{x_{i} - x_{j}}_{2}^{2}.
\end{align*}
\end{proof}

\subsection{An alternative proof}

We can alternatively get the same guarantee as in \cref{thm:stable-rank}, by giving a one-dimensional $\ell_2$ embedding (and hence also $\ell_1$ embedding without any extra effort) along the top singular vector of the difference matrix $M$. This gives an interesting ``spectral'' algorithm that uses spectral information about the point set, akin to spectral algorithms in graphs that use the spectrum of the graph Laplacian.

\begin{theorem} \label{thm:stable-rank-2}
Let $x_{1}, x_{2}, \dotsc, x_{n} \in \R^{d}$ be $n$ points satisfying $\ell_{2}^{2}$ triangle inequalities with $M$ as their difference matrix. Let $u \in \R^{d}$ and $v \in \R^{{n \choose 2}}$ be its top left and right singular vectors, respectively. Then $x_{i} \mapsto \frac{\sigma_{1}}{\norm{v}_{1}} \inner{x_{i}}{u}$ is an $\ell_{2}^{2}\hookrightarrow\ell_{2}$ embedding with average distortion bounded by the stable rank of $M$.
\end{theorem}
\begin{proof}
We have $Mv = \sigma_{1} u$, or equivalently, $\sigma_{1} u = \sum_{k<l} v_{kl} (x_{k} - x_{l})$. Our embedding is a contraction since
\begin{align*}
\frac{\sigma_{1}}{\norm{v}_{1}} \abs{\inner{x_{i} - x_{j}}{u}} & = \frac{1}{\norm{v}_{1}} \abs{\inner{x_{i} - x_{j}}{\sum_{k<l} v_{kl} (x_{k} - x_{l})}} \\
& \leq \frac{1}{\norm{v}_{1}} \sum_{k<l} \abs{v_{kl}} \abs{\inner{x_{i} - x_{j}}{x_{k} - x_{l}}} \\
& \leq \frac{1}{\norm{v}_{1}} \sum_{k<l} \abs{v_{kl}} \norm{x_{i} - x_{j}}_{2}^{2} &[\text{By {\cref{lemma:key}}}]\\
& = \norm{x_{i} - x_{j}}_{2}^{2}.
\end{align*}
Now let's bound the average distortion.
\begin{align*}
\sum_{i<j} \frac{\sigma_{1}}{\norm{v}_{1}} \abs{\inner{x_{i} - x_{j}}{u}} & = \sum_{i<j} \frac{\sigma_{1}}{\norm{v}_{1}} \abs{\sigma_{1} v_{ij}} &[\text{Since $u^{T}M = \sigma_{1} v^{T}$}]\\
& = \sigma_{1}^{2} = \frac{\fnorm{M}^2}{\sr{M}}\\
& = \frac{1}{\sr{M}}~ \sum_{i<j} \norm{x_{i} - x_{j}}_{2}^{2}.
\end{align*}
\end{proof}

\subsection{Application to \sparsestcut on low-threshold rank graphs}

We first state a property of  SDP solutions on low threshold-rank graphs, proved by Guruswami and Sinop~\cite{GuruswamiS2013} using the Von-Neumann inequality.

\begin{prop}[Von-Neumann inequality~{\cite[Theorem~3.3]{GuruswamiS2013}}]\label{prop:laplacian}
Let $0 \leq \lambda_{1} \leq \dotsc \leq \lambda_{m}$ be the generalized
eigenvalues of the Laplacian matrices of the cost and demand graphs. Let
$\sigma_{1} \geq \sigma_{2} \geq \dotsc \geq \sigma_{n} \geq 0$ be the singular vectors of the matrix $M$ with columns $\{\sqrt{d_{ij}} (x_{i} - x_{j})\}_{i<j}$. Then
\[
\frac{\sum_{t \geq r+1} \sigma_{j}^2}{\sum_{t=1}^{n} \sigma_{j}^2} \leq
\frac{\Phi_{SDP}}{\lambda_{r+1}}.
\]
\end{prop}
In particular, note that on graphs where $\lambda_r \geq \Phi_{SDP}/(1-\delta)$, $\sum_{i \leq r} \sigma_{i}^2 \geq \delta \sum_{i} \sigma_i^2 $. This implies that $\sr{M} =\sum_i \sigma_i^2/\sigma_1^2 \leq r ~ \cdot ~\sum_i \sigma_i^2/ \sum_{i\leq r} \sigma_i^2 \leq r/\delta$.

We can now modify the proof of \cref{thm:stable-rank-2} to prove \cref{thm:SparsestCut}.


\begin{proof}[Proof of {\cref{thm:SparsestCut}}]
Let $x_1, \dotsc , x_n$ be the SDP solution on given instance $C, D$.
We now let $M$ be the matrix with columns $\{\sqrt{d_{kl} }(x_k-x_l)\}_{k<l}$, and  $u,v,\sigma_1$ to be the top left singular vector, top right singular vector, and the maximum singular value respectively of $M$. By the preceding remark, $\sr{M}\leq r/\delta$. The mapping we use is as follows
\[ x_i \mapsto \frac{1}{\sum_{kl} \sqrt{d_{kl}} v_{kl}}\inner{x_i}{u}
\]
The proofs to show contraction and bound the distortion follow exactly as in the proof of \cref{thm:stable-rank-2}. Note that while looking at the distortion, we need to lower bound $\sum_{ij}  d_{ij} \norm{g(x_i)-g(x_j)}_2$.
\end{proof}

As in Deshpande and Venkat~\cite{DeshpandeV2014}, the above algorithm is a fixed polynomial time algorithm and does not grow with the threshold rank unlike the algorithm of Guruswami and Sinop~\cite{GuruswamiS2013} where they use $r$-levels of the Lasserre SDP hierarchy to secure the guarantee. Furthermore, the above analysis improves the guarantee of Deshpande and Venkat~\cite{DeshpandeV2014} by a factor of $O(1/\delta)$.

\section{Embedding low-dimensional vectors \`{a} la Goemans} \label{sec:goemans}

In this section, we first view the proof of Goemans' theorem in the
framework of \cref{lem:contraction} by giving a
probability distribution using the minimum volume enclosing elliposid   
of the difference vectors $(x_{i}-x_{j})$'s. We then give a simpler
proof, albeit for the {\em average} distortion case, based on a
probability distribution arising from squared-length distribution. Via a well-known duality statement, this technique recovers Goemans' theorem for {\em worst-case} distortion for embeddings into $\ell_1$, although non-constructively. 
 
\subsection{An alternate proof of Goemans' theorem}
Here is an adaptation of the proof from \cite{MagenM2008} re-stated in
our framework. The following proof is arguably simpler and more
straightforward as it works with the difference vectors instead of the
original vectors and their negations.\\

\noindent {\bf \cref{thm:goemans-intro}} (Restated) (Goemans~\cite[Appendix~B]{MagenM2008})
{\em  Let $x_{1}, x_{2}, \dotsc, x_{n} \in \R^{d}$ be $n$ points satisfying $\ell_{2}^{2}$
triangle inequalities. Then there exists an $\ell_{2}^{2}\hookrightarrow\ell_{2}$
embedding $x_{i} \mapsto f(x_{i})$ with distortion $\sqrt{d}$, that is,
\[
\frac{1}{\sqrt{d}}~ \norm{x_{i} - x_{j}}_{2}^{2} \leq \norm{f(x_{i}) -
f(x_{j})}_{2} \leq \norm{x_{i} - x_{j}}_{2}^{2}, \quad \forall~ i, j\in
V.
\]}

\begin{proof}
Consider all the difference vectors $(x_{i} - x_{j})$'s, and let their
minimum volume enclosing ellipsoid be given by $E := \{x \suchthat
x^{T} Q x \leq 1\}$, for some psd matrix $Q \in \R^{d \times d}$. By John's theorem (or Lagrangian duality for the corresponding convex program), we have $Q^{-1} = \sum_{k<l} \alpha_{kl}~ (x_{k} - x_{l})(x_{k} - x_{l})^{T}$, with all $\alpha_{kl} \geq 0$. Moreover, $\alpha_{kl} \neq 0$ iff $(x_{k} - x_{l})^{T} Q (x_{k} - x_{l}) = 1$. Notice that $d = \tr{\mathbb{I}_{d}} = \tr{Q^{1/2}  Q^{-1}Q^{1/2}} = \sum_{k<l} \alpha_{kl}$. We define the embedding as
\[
f(x_{i}) := \frac{1}{\sqrt{d}} Q^{-1/2} x_{i}.
\]
This embedding is a contraction by \cref{lem:contraction}. We now bound the distortion:
\begin{align*}
\norm{f(x_{i}) - f(x_{j})}_{2} & = \frac{1}{\sqrt{d}}~ \norm{Q^{-1/2} (x_{i} - x_{j})}_{2} \\
& \geq \frac{1}{\sqrt{d}}~ \frac{\norm{x_{i} -
  x_{j}}_{2}^{2}}{\norm{Q^{1/2} (x_{i} - x_{j})}_{2}} & [\text{By Cauchy-Schwarz inequality}] \\
& \geq \frac{1}{\sqrt{d}}~ \norm{x_{i} - x_{j}}_{2}^{2}. &[\text{Since $(x_{i} - x_{j})^{T} Q (x_{i} - x_{j}) \leq 1$, for all $i, j$}]
\end{align*}
\end{proof}


\subsection{A simpler proof for average distortion embedding} \label{sec:average}
We now give an average distortion version of Goemans' theorem using a simple squared-length distribution on the difference vectors $(x_{i} - x_{j})$'s in the \cref{lem:contraction}. Interestingly, this can be modified to weighted averages and gives yet another proof of Goemans' worst-case distortion result, although non-constructively.

\begin{theorem} \label{thm:average}
Let $x_{1}, x_{2}, \dotsc, x_{n} \in \R^{d}$ be points satisfying $\ell_{2}^{2}$ triangle inequalities. Then there exists an $\ell_{2}^{2}$-to-$\ell_{2}$ embedding $x_{i} \mapsto g(x_{i})$ with \emph{average} distortion $\sqrt{d}$, that is,
\begin{align*}
\norm{g(x_{i}) - g(x_{j})}_{2} & \leq \norm{x_{i} - x_{j}}_{2}^{2}, \quad \text{for all $i, j$}, \\
\text{and} \qquad \frac{1}{\sqrt{d}}~ \sum_{i<j} \norm{x_{i} - x_{j}}_{2}^{2} & \leq \sum_{i<j} \norm{g(x_{i}) - g(x_{j})}_{2}
\end{align*}
\end{theorem}
\begin{proof}
Let $\{p_{kl}\}_{k<l}$ define a probability distribution with $p_{kl} \propto \norm{x_{k} - x_{l}}_{2}^{2}$. Given this distribution, let $P$ be the symmetric psd matrix defined as $P := \sum_{k<l} p_{kl}~ (x_{k} - x_{l}) (x_{k} - x_{l})^{T} \in \R^{d \times d}$. Consider the embedding that maps $x_{i}$ to $g(x_{i}):= P^{1/2} x_{i}$. The embedding is a contraction by the \cref{lem:contraction}.

Now let's bound the average distortion. First, note that:
\begin{align*}
\norm{g(x_{i}) - g(x_{j})}_{2} & =\norm{P^{1/2} (x_{i} - x_{j})}_{2} \geq \frac{\norm{x_{i} - x_{j}}_{2}^{2}}{\norm{P^{-1/2} (x_{i} - x_{j})}_{2}},
\end{align*}
where the inequality follows from the Cauchy-Schwarz inequality.

Summing over all pairs $i,j$ and using the definition of $p_{ij}$ we have
\begin{align*}
\sum_{i<j}\norm{g(x_i)-g(x_j)}_2 & \geq \left( \sum_{k<l} \norm{x_{k} - x_{l}}_{2}^{2}\right)~ \sum_{i<j} \frac{p_{ij}}{\sqrt{(x_{i} - x_{j})^{T} P^{-1} (x_{i} - x_{j})}} \\
& \geq ~ \left( \sum_{k<l} \norm{x_{k} - x_{l}}_{2}^{2}\right) \left(\sum_{i<j} p_{ij}~ (x_{i} - x_{j})^{T} P^{-1} (x_{i} - x_{j})\right)^{-1/2} \\
& \qquad & [\text{by Jensen's inequality}] \\
& = \left( \sum_{k<l} \norm{x_{k} - x_{l}}_{2}^{2}\right)~ \left( \tr{ P^{-1/2} PP^{-1/2}}^{-1/2}\right) \\
& = \left( \sum_{k<l} \norm{x_{k} - x_{l}}_{2}^{2}\right)~ \tr{\mathbb{I}_{d}}^{-1/2} \\
& = \frac{1}{\sqrt{d}}~ \sum_{i<j} \norm{x_{i} - x_{j}}_{2}^{2}.
\end{align*}
We note that if $P$ is not invertible then the same proof can be carried out using pseudo-inverse of $P$ instead.
\end{proof}

\begin{remark}
We note that our embedding (for the average distortion case) is based on a simple squared-length distribution and does not involve computation of minimum volume enclosing ellipsoid \cite{Khachiyan1996} as in the earlier proof.
\end{remark}
\Cref{thm:average} immediately gives an efficient $\sqrt{d}$ approximation algorithm for \uniformSparsestcut when the SDP optimum solution resides in $\R^d$. Furthermore, as we point out next, the same proof can be tweaked to yield a similar result for the general \sparsestcut problem.

\begin{theorem}[{\sparsestcut} SDP rounding in dimension $d$]\label{thm:d-dimCut}
A \sparsestcut instance $C,D$ with SDP optimum solution in $\R^d$ has an integrality gap of at most $\sqrt{d}$.
\end{theorem}
\begin{proof}
Let ${x_1, \dotsc x_n}$ be the optimum solution in $\R^d$ to the \sparsestcut SDP. We slightly modify the embedding given in the proof of \cref{thm:average}, by choosing the $p_{ij}$'s based on the demand graph $D$. Let $P = \sum_{k<l} p_{kl}~ (x_{k} - x_{l}) (x_{k} - x_{l})^{T} \in \R^{d \times d}$, where $p_{kl}$'s define a probability distribution with $p_{kl} \propto d_{kl}\norm{x_{k} - x_{l}}_{2}^{2}$. We define the embedding as $x_{i} \mapsto g(x_{i}) = P^{1/2} x_{i}$. \Cref{lem:contraction} shows that it is a contraction. We now need to show $\sum_{i<j} d_{ij} \norm{g(x_i)-g(x_j)}_2 \geq \frac{1}{\sqrt{d}} \sum_{i<j}d_{ij}\norm{x_i-x_j}_2^2$. It is easy to check that the same proof goes through without any major changes.
\end{proof}
By a well-known duality (cf. \cite[Proposition~15.5.2  and Exercise~4]{Matousek2002}), \cref{thm:d-dimCut} also implies Goemans' worst-case distortion result (\cref{thm:goemans-intro}), although non-constructively.


\subsection*{Acknowledgements}
We thank Luca Trevisan for helpful discussions and suggestions, in particular, for bringing to our attention that Goemans' Theorem was true even with approximate triangle inequalities.

{\small
\bibliographystyle{plainurl}
\bibliography{DHV-bib}
}

\end{document}